\documentclass[a4paper,11pt]{article}

\usepackage{fullpage}
\usepackage{amsthm}
\usepackage{amssymb}
\usepackage{amsmath}
\usepackage{graphicx}
\usepackage{algorithm}
\usepackage[noend]{algorithmic}
\usepackage{epstopdf}

\newtheorem{theorem}{Theorem}[section]
\newtheorem{lemma}[theorem]{Lemma}

\author{
  Iffat Chowdhury and Matt Gibson\\
	Department of Computer Science\\
University of Texas at San Antonio\\
San Antonio, TX, USA\\
  \texttt{iffat.chowdhury@utsa.edu, gibson@cs.utsa.edu}
}
\title{A Characterization of Consistent Digital Line Segments in $\mathbb{Z}^2$}

\begin{document}
\bibliographystyle{plain}
  \maketitle
    \begin{abstract}
Our concern is the digitalization of line segments in $\mathbb{Z}^2$ as considered by Chun et al.~\cite{ChunKNT09} 
and Christ et al.~\cite{ChristPS12}.  
The key property that differentiates the research of Chun et al.~and Christ et al.~from other research in digital line segment construction is that the intersection of any two segments must be connected.  Such a system of segments is called a consistent digital line segments system (CDS). Chun et al.~give a construction for all segments in $\mathbb{Z}^d$ that share a common endpoint (called consistent digital rays (CDR)) that has asymptotically optimal Hausdorff distance, and Christ et al.~give a complete CDS in $\mathbb{Z}^2$ with optimal Hausdorff distance.  Christ et al.~also give a characterization of CDRs in $\mathbb{Z}^2$, and they leave open the question on how to characterize CDSes in $\mathbb{Z}^2$.  In this paper, we answer the most important open question regarding CDSes in $\mathbb{Z}^2$ by giving the characterization asked for by Christ et al.  We obtain the characterization by giving a set of necessary and sufficient conditions that a CDS must satisfy.
    \end{abstract}

\section{Introduction}
\label{sec:intro}
This paper explores families of digital line segments as considered by Chun et al.~\cite{ChunKNT09} and Christ et al.~\cite{ChristPS12}. Consider the unit grid $\mathbb{Z}^2$, and in particular the unit grid graph: for any two points $p=(p^x,p^y)$ and $q=(q^x,q^y)$ in $\mathbb{Z}^2$, $p$ and $q$ are neighbors if and only if $|p^x-q^x|+|p^y-q^y|=1$.  For any pair of grid vertices $p$ and $q$, we'd like to define a digital line segment $R_p(q)$ from $p$ to $q$.  The collection of digital segments must satisfy the following five properties.

\begin{description}
\item $(S1)$ \textit{Grid path property:} For all $p, q \in \mathbb{Z}^2$, $R_p(q)$ is the points of a path from $p$ to $q$ in the grid topology.
\item $(S2)$ \textit{Symmetry property:} For all $p, q \in \mathbb{Z}^2$, we have $R_p(q) = R_q(p)$.
\item $(S3)$ \textit{Subsegment property:} For all $p, q \in \mathbb{Z}^2$ and every $r,s \in R_p(q)$, we have $R_r(s) \subseteq R_p(q)$.
\end{description} 

Properties $(S2)$ and $(S3)$ are quite natural to ask for; the subsegment property $(S3)$ is motivated by the fact that the intersection of any two Euclidean line segments is connected. See Fig.~\ref{fig:property} (a) for an illustration of a violation of (S3).  Note that a simple ``rounding'' scheme of a Euclidean segment commonly used in computer vision produces a good digitalization in isolation, but unfortunately it will not satisfy (S3) when combined with other digital segments, see Fig.~\ref{fig:property} (b) and (c).   

\begin{description}
 \item $(S4)$ \textit{Prolongation property:} For all $p, q \in \mathbb{Z}^2$, there exists $r \in \mathbb{Z}^2$, such that $r \notin R_p(q)$ and $R_p(q) \subseteq R_p(r)$.
\end{description}

The prolongation property $(S4)$ is also a quite natural property to desire with respect to Euclidean line segments.  Any Euclidean line segment can be extended to an infinite line, and we would like a similar property to hold for our digital line segments.  While (S1)-(S4) form a natural set of axioms for digital segments, there are pathological examples of segments that satisfy these properties which we would like to rule out. For example, Christ et al.~\cite{ChristPS12} describe a CDS where a double spiral is centered at some point in $\mathbb{Z}^2$, traversing all points of $\mathbb{Z}^2$.  A CDS is obtained by defining $R_p(q)$ to be the subsegment of this spiral connecting $p$ and $q$.   To rule out these CDSes, the following property was added.

\begin{description}
\item $(S5)$ \textit{Monotonicity property:} For all $p,q \in \mathbb{Z}^2$, if $p^x = q^x = c_1$ for any $c_1$ (resp. $p^y = q^y = c_2$ for any $c_2$), then every point $r \in R_p(q)$ has $r^x=c_1$ (resp. $r^y=c_2$).
\end{description}

\begin{figure}[htpb]
\centering
\begin{tabular}{c@{\hspace{0.01\linewidth}}c@{\hspace{0.01\linewidth}}c}

\includegraphics[width=1.4in]{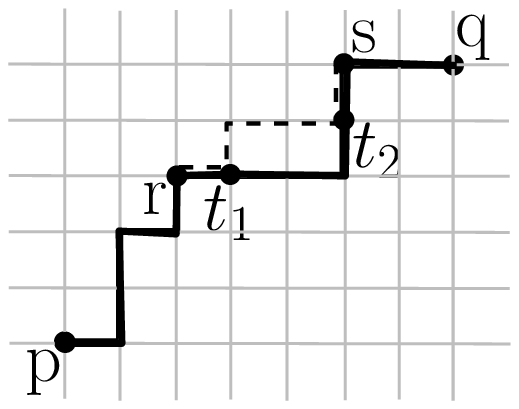} &
\includegraphics[width=1.3in]{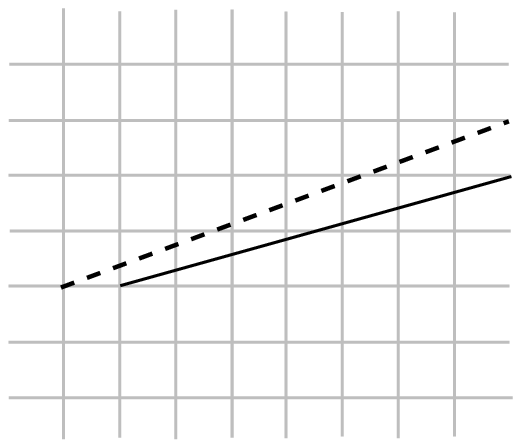} &
\includegraphics[width=1.3in]{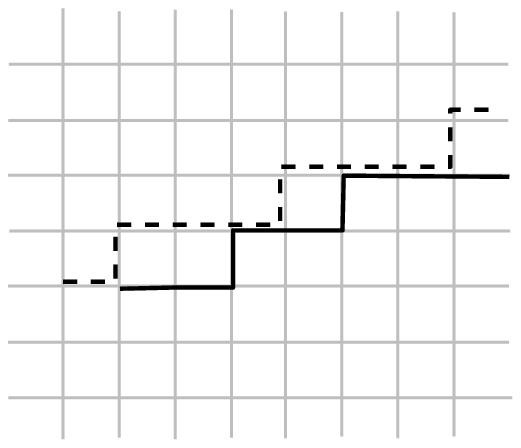}\\

(a) & (b) & (c)
\end{tabular}
\caption{(a) An illustration of the violation of (S3). The solid segment is $R_p(q)$, and the dashed segment is $R_r(s)$.  (b) The dashed line and the solid line denote two different Euclidean line segments. (c) The corresponding digital line segments via a rounding approach.}
\label{fig:property}
\end{figure}

If a system of digital line segments satisfies the axioms $(S1)-(S5)$, then it is called a \textit{consistent digital line segments system} (CDS).  Given such a system, one can easily define digital analogs of various Euclidean objects.  For example, a Euclidean object $O$ is convex if for any two points $p,q \in O$ we have that the Euclidean line segment $\overline{pq}$ does not contain any points outside of $O$.  Given a CDS, the natural definition of a digital convex object will satisfy some nice properties.  For example, one can see that a digital convex object with respect to a CDS cannot contain any holes (as a result of the prolongation property $(S4)$).  Similarly, one can easily obtain the digital analog of a star-shaped object with a CDS.  A Euclidean object $S$ is star-shaped if there is a point $u \in S$ such that for every point $v \in S$ we have the Euclidean line segment $\overline{u v}$ does not contain any points outside of $S$, and the natural digital generalization easily follows.

\textbf{Previous Works.} 
Unknown to Chun et al.~and Christ et al.~when publishing their papers, Luby \cite{Luby88} considers \textit{grid geometries} which are equivalent to systems of digital line segments satisfying (S1), (S2), (S5) described in this paper.  Let $p,q \in \mathbb{Z}^2$ be such that $p^x \leq q^x$.  We say the segment $R_p(q)$ has \textit{nonnegative slope}  if $p^y \leq q^y$ and otherwise has \textit{negative slope}.  Luby investigates a property called \textit{smoothness} which uses the following notion of distance between two digital line segments. Consider two digital segments $R_p(q)$ and $R_{p'}(q')$ with nonnegative slope and any point $r \in R_{p}(q)$.  If there is a $s\in R_{p'}(q')$ such that $r^x + r^y = s^x + s^y$, then dist$(R_p(q), R_{p'}(q'),r)$ is defined to be $r^x - s^x$.  If there is no such $s$ then we say that dist$(R_p(q), R_{p'}(q'),r)$ is undefined.  See Fig.~\ref{fig:smooth} (a).  The segments are \textit{smooth} if dist$(R_p(q), R_{p'}(q'),r)$ is either monotonically increasing or monotonically decreasing varying $r$ over its defined domain. See Fig.~\ref{fig:smooth} (b).  There is also a symmetric definition of smoothness for pairs of digital segments with negative slope.  A grid geometry is said to be smooth if every pair of nonnegative sloped segments and every pair of negative sloped segments are smooth.  Luby shows that if a grid geometry is smooth, then it satisfies properties (S3) and (S4) (and therefore is a CDS).

\begin{figure}[htpb]
\centering

\begin{tabular}{c@{\hspace{0.01\linewidth}}c}

\includegraphics[width=1.3in]{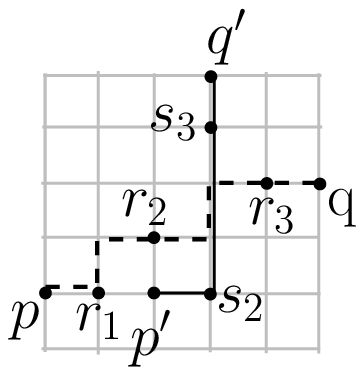} &
\includegraphics[width=1.15in]{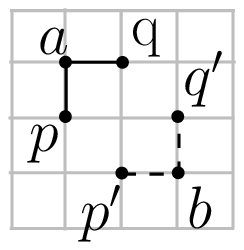}\\

(a) & (b)
\end{tabular}
\caption{(a) dist$(R_p(q), R_{p'}(q'),r_1)$ is undefined.  dist$(R_p(q), R_{p'}(q'),r_2) = -1$ (using $s_2 \in R_{p'}(q'))$, dist$(R_p(q), R_{p'}(q'),r_3) = 1$.  (b) An example of segments that are not smooth: dist$(R_p(q), R_{p'}(q'),p) = -1$, dist$(R_p(q), R_{p'}(q'),a) = -2$, and dist$(R_p(q), R_{p'}(q'),q) = -1$.}
\label{fig:smooth}
\end{figure}

Chun et al.~\cite{ChunKNT09} give an  $\Omega(\log n)$ lower bound on the Hausdorff Distance of a CDS where $n$ is the number of points in the segment, and the result even applies to \textit{consistent digital rays} or CDRs (i.e., all segments share a common endpoint).  Note that this lower bound is due to property (S3), as it is easy to see that if the requirement of (S3) is removed then digital segments with $O(1)$ Hausdorff distance are easily obtained, for example the trivial ``rounding'' scheme used in Fig.~\ref{fig:property} (c).  Chun et al.~give a construction of CDRs that satisfy the desired properties (S1)-(S5) with a tight upper bound of $O(\log n)$ on the Hausdorff distance.  Christ et al.~\cite{ChristPS12} extend the result to get an optimal $O(\log n)$ upper bound on Hausdorff distance for a CDS in $\mathbb{Z}^2$. 

After giving the optimal CDS in $\mathbb{Z}^2$, Christ et al.~\cite{ChristPS12} investigate common patterns in CDSes in an effort to obtain a characterization of CDSes. As a starting point, they are able to give a characterization of CDRs. In their effort to give a characterization, they proved a sufficient condition on the construction of the CDSes but then they give an example of a CDS that demonstrates that their sufficient condition is not necessary. They ask if there are any other interesting examples of CDSes that do not follow their sufficient condition and left open the question on how to characterize the CDSes in $\mathbb{Z}^2$.

\textbf{Our Contributions.} In this paper, we answer the most important open question regarding CDSes in $\mathbb{Z}^2$ by giving the characterization asked for by Christ et al.  Since Christ et al.~has given a characterization of CDRs, we view the construction of a CDS as the assignment of a CDR system to each point in $\mathbb{Z}^2$ such that the union of these CDRs satisfies properties (S1)-(S5).  We obtain the characterization by giving a set of necessary and sufficient conditions that the CDRs must satisfy in order to be combined into a CDS.  Then to tie together our work with the previous work in CDSes, we analyze the work of Christ et al.~and Luby in the context of our characterization. 

\textbf{Motivation and Related Works.} Digital geometry plays a fundamental and substantial role in many computer vision applications, for example image segmentation, image processing, facial recognition, fingerprint recognition, and some medical applications. One of the key challenges in digital geometry is to represent Euclidean objects in a digital space so that the digital objects have a similar visual appearance as their Euclidean counterparts.  Representation of Euclidean objects in a digital space has been a focus in research for over $25$ years, see for example \cite{Franklin85,GreeneY86,Sugihara01,LieropOW88,Eckhardt00,BertheL11,Andres03}. 

Digital line segments are particularly important to model accurately, as other digital objects depend on them for their own definitions (e.g. convex and star-shaped objects). In $1986$, Greene and Yao \cite{GreeneY86} gave an interface between the continuous domain of Euclidean line segments and the discrete domain of digital line segments. Goodrich et al.~\cite{GoodrichGHT97} focused on rounding the Euclidean geometric objects to a specific resolution for better computer representation. They gave an efficient algorithm for $\mathbb{R}^2$ and $\mathbb{R}^3$ in the ``snap rounding paradigm'' where the endpoints or the intersection points of several different line segments are the main concerns. Later, Sivignon et al.~\cite{SivignonDC04} also gave some results on the intersection of two digital line segments. In their review paper, Klette et al.~\cite{KletteR04} discussed the straightness of digital line segments. The characteristics of the subsegment of digital straight line was computed in \cite{LachaudS13}. Cohen et al.~\cite{Cohen-OrK97} gave a method of converting 3D continuous line segments to discrete line segments based on a voxelization algorithm, but they did not have the requirement that the intersection of two digital line segments should be connected. 

\section{Preliminaries}
\label{sec:prelim}
Before we describe our characterization, we first need to give some details of the Christ et al.~characterization of CDRs.  For any point $p\in \mathbb{Z}^2$, let $Q^1_p, Q^2_p, Q^3_p, Q^4_p$ denote the first, second, third, and fourth quadrants of $p$ respectively.  Christ et al. show how to construct $R_p(q)$ for $q \in Q^1_p$ from any total order of $\mathbb{Z}$, which we denote $\prec^1_p$.  We describe $R_p(q)$ by ``walking'' from $p$ to $q$. Starting from $p$, the segment will move either ``up'' or ``right'' until it reaches $q$.  Suppose on the walk we are currently at a point $r=(r^x,r^y)$. Then it needs to move to either $(r^x+1,r^y)$ or $(r^x,r^y+1)$. Either way, the sum of the two coordinates of the current point is increased by $1$ in each step. The segment will move up $q^y-p^y$ times, and it will move right $q^x-p^x$ times. If the line segment is at a point $r$ for which $r^x+r^y$ is among the $q^y-p^y$ greatest integers in the interval $I(p,q) :=[p^x+p^y, q^x+q^y-1]$ according to $\prec^1_p$, the line segment will move up. Otherwise, it will move right. See Fig.~\ref{fig:charFig1} for an example. Throughout the paper, when we say that $a < b$, we mean that $a$ is less than $b$ in natural total order and when we say that $a \prec b$, we mean that $a$ is less than $b$ according to total order $\prec$. 

\begin{figure}[htpb]
\centering
\begin{tabular}{c@{\hspace{0.01\linewidth}}c@{\hspace{0.01\linewidth}}c}
\includegraphics[width=2in]{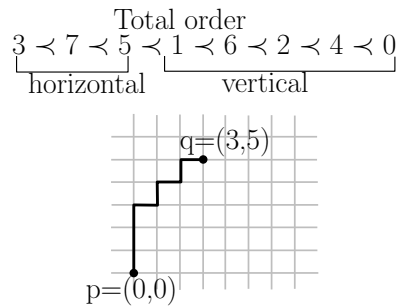} \\

\end{tabular}
\caption{The digital line segment between $p=(0,0)$ and $q=(3,5)$. According to $\prec_p$, the $q^x-p^x = 3$ smallest integers in $[0, 7]$ correspond to the horizontal movements, and the $q^y-p^y = 5$ largest integers in $[0, 7]$ correspond to the vertical movements.}
\label{fig:charFig1}
\end{figure}

Property (S3) is generally the most difficult property to deal with, and we will argue that the segments $R_p(q)$ and $R_p(q')$ will not violate (S3) for any points $q$ and $q'$ in the first quadrant of $p$.  As shown in \cite{ChristPS12}, (S3) is violated if and only if two segments intersect at a point $t_1$, one segment moves vertically from $t_1$ while the other moves horizontally from $t_1$, and the segments later intersect again.  Consider two digital segments that ``break apart'' at some point $t_1$ in this manner, and suppose they do intersect again.  Let $t_2$ be the first point at which they intersect after ``splitting apart''.  Then we say that $(t_1,t_2)$ is \textit{witness to the violation of (S3)} or a \textit{witness} for short. Therefore, one can show that any two segments satisfy $(S3)$ by showing that they do not have witnesses, and this is how we will prove the segments satisfy (S3) now (and also in our characterization). Consider the segments $R_p(q)$ and $R_p(q')$ generated according to the Christ et al.~definition, and suppose for the sake of contradiction that they have a witness $(t_1,t_2)$ as in Fig.~\ref{fig:property} (a).  One segment moves up at point $t_1$ and moves right into the point $t_2$ which implies $(t_2^x + t_2^y - 1) \prec^1_p (t_1^x + t_1^y)$, and the other segment moves right at point $t_1$ and moves up into the point $t_2$ which implies $(t_1^x + t_1^y) \prec^1_p (t_2^x + t_2^y - 1)$, a contradiction.  Therefore $R_p(q)$ and $R_p(q')$ do not have any witnesses and therefore satisfy (S3).  Christ et al.~\cite{ChristPS12} show that digital segments in quadrants $Q^2_p$, $Q^3_p$, and $Q^4_p$ can also be generated with total orders $\prec^2_p, \prec^3_p$, and $\prec^4_p$ (described formally below), and moreover they establish a one-to-one correspondence between CDRs and total orders.  That is, (1) given any total order of $\mathbb{Z}$, one can generate all digital rays in any quadrant of $p$, and (2) for any set of digital rays $R$ in some quadrant of $p$, there is a total order that will generate $R$.  This provides a characterization of CDRs.

Given the characterization of CDRs, the problem of constructing a complete CDS can be viewed as assigning total orders to all points in $\mathbb{Z}^2$ so that the segments obtained using these total orders are collectively a CDS.  Suppose that for every point $p \in \mathbb{Z}^2$, we assign to $p$ a total order $\prec^1_p$ to generate segments to all points in $Q^1_p$. Now suppose that we want to define a ``third-quadrant segment'' $R_p(q)$ to some point $q \in Q^3_p$.  Note that $q \in Q^3_p$ implies that $p \in Q^1_{q}$.  Since all first-quadrant segments have been defined, this means $R_{q}(p)$ has been defined, and the symmetry property (S2) states that $R_p(q) = R_{q}(p)$.  Therefore we do not need to use a total order to generate these third-quadrant segments; we simply use the corresponding first-quadrant segments which have already been defined. In order to be part of a CDS, $\bigcup_{q \in Q^3_p} R_q(p)$ must be a system of rays in $Q^3_p$ that satisfies (S1)-(S5).  From the characterization of rays, we know that there is an implicit third-quadrant total order $\prec^3_p$ on the integers in the range $(-\infty, p^x+p^y]$ that can be used to generate these rays.  This generation is done in a very similar manner as in first quadrant rays.  The key differences are: (1) the first quadrant segment $R_q(p)$ uses the interval $[q^x + q^y, p^x + p^y -1]$ and the third quadrant segment $R_p(q)$ uses the interval $[q^x + q^y + 1, p^x + p^y]$, and (2) the sum of the coordinates of our ``current point'' \textit{decreases} by 1 each time as we walk from $p$ to $q$.  Note that when considering first-quadrant segments, a horizontal movement (resp. vertical movement) is determined by the sum of the coordinates of the ``left'' endpoint (resp. ``bottom'' endpoint), whereas in a third quadrant segment a horizontal  movement (resp. vertical movement) is determined by the sum of the coordinates of the ``right'' endpoint (resp. ``top'' endpoint).  This implies that if the first-quadrant segment made a horizontal (resp. vertical) movement at a point where the sum of the coordinates is $a$, then the corresponding third-quadrant segment should make a horizontal (resp. vertical) movement at $(a+1)$.  For example consider Fig.~\ref{fig:charFig1}.  Since the horizontal movements of this first quadrant segment are at 3, 7, and 5, then the third quadrant segment should make horizontal movements at 4, 8, and 6.  Similarly, the third quadrant segment should make vertical movements at 2, 7, 3, 5, and 1. We again state that we do not explicitly construct third quadrant segments using this technique and instead obtain them directly from the corresponding first quadrant segments.  But note that if there is no total order $\prec^3_p$ which can be used to generate these third-quadrant segments then the segments necessarily must not satisfy at least one of (S1)-(S5). These implicit third-quadrant total orders will play an important role in the proof of our characterization.

Now consider the definition of segments $R_p(q)$ with negative slope, that is, $R_p(q)$ for which $q \in Q^2_p$ or $q \in Q^4_p$. We ``mirror'' $p$ and $q$ by multiplying both $x$-coordinates by $-1$.  Let $m(p) = (-p^x, p^y)$ and $m(q) = (-q^x, q^y)$ denote the mirrored points.  Note that if $q \in Q^2_p$, then $m(q) \in Q^1_{m(p)}$, and if $q \in Q^4_p$, then $m(q) \in Q^3_{m(p)}$.  Therefore $R_{m(p)}(m(q))$ is a segment with nonnegative slope and can be defined as described above.  We compute a second-quadrant segment $R_p(q)$ by making the same sequence of horizontal/vertical movements as $R_{m(p)}(m(q))$ when generated by a second-quadrant total order $\prec^2_p$ on the integers in the range $[-p^x+p^y, \infty)$.  Similarly to third quadrant segments, fourth-quadrant segments $R_p(q)$ are set to be the same as $R_q(p)$ and there is an implicit fourth-quadrant total order $\prec^4_p$ on the integers in the range $(-\infty, -p^x+p^y]$.


\section{A Characterization of CDSes in $\boldmath{Z}^2$}
\label{sec:necessary} 
In a complete CDS, the segments that are adjacent to any point $p\in\mathbb{Z}^2$ can be viewed as a system of CDRs emanating from $p$, and therefore there is a total order than can be used to generate these segments.  Christ et al.~show that if the same total order is used by every point in $\mathbb{Z}^2$ to generate its adjacent segments, then the result is a CDS (the analysis follows very closely to the analysis for CDRs shown in the previous section).  However, there are some situations in which points can be assigned different total orders and we still get a CDS.  To illustrate this, consider Fig.~\ref{fig:charFig} (a).  Note that $\prec_p^1$ and $\prec_{p'}^1$ disagree on the relative ordering of 4 and 6, yet the resulting segments $R_p(q)$ and $R_{p'}(q')$ satisfy property (S3).  But if we instead use the total orders as shown in Fig.~\ref{fig:charFig} (b), they once again disagree on the ordering of 4 and 6 but this time $R_p(q)$ and $R_{p'}(q')$ do not satisfy property (S3).  The issue is then to identify a set of necessary and sufficient properties of the total orders in a CDS.

\begin{figure}[htpb]
\centering
\begin{tabular}{c@{\hspace{0.01\linewidth}}c@{\hspace{0.01\linewidth}}c}
\includegraphics[width=2.3in]{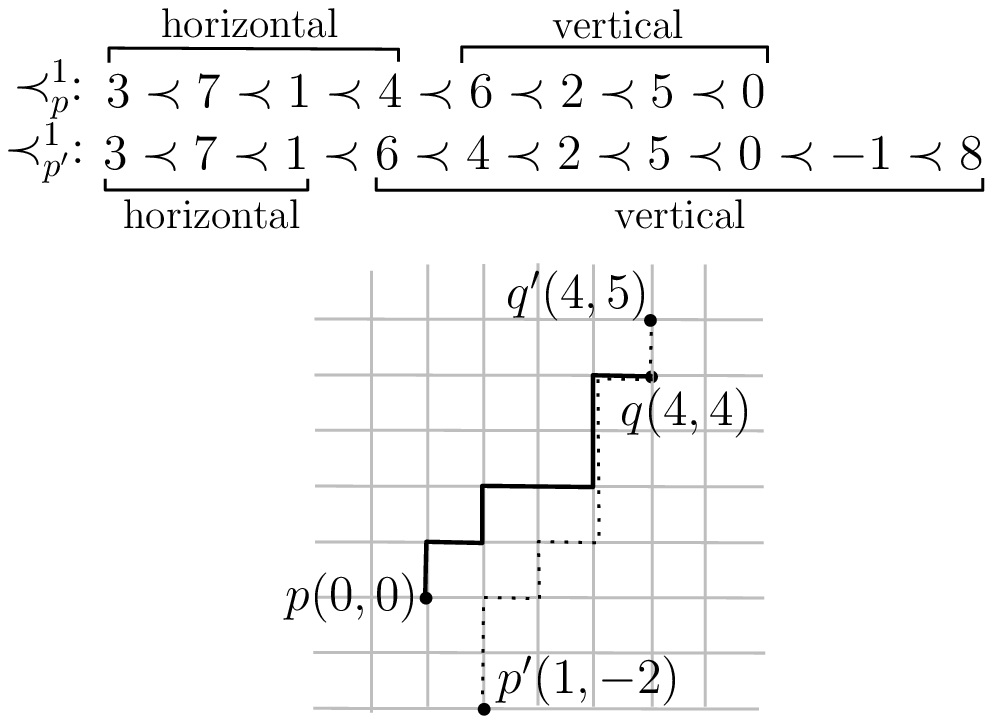} &
\includegraphics[width=2.3in]{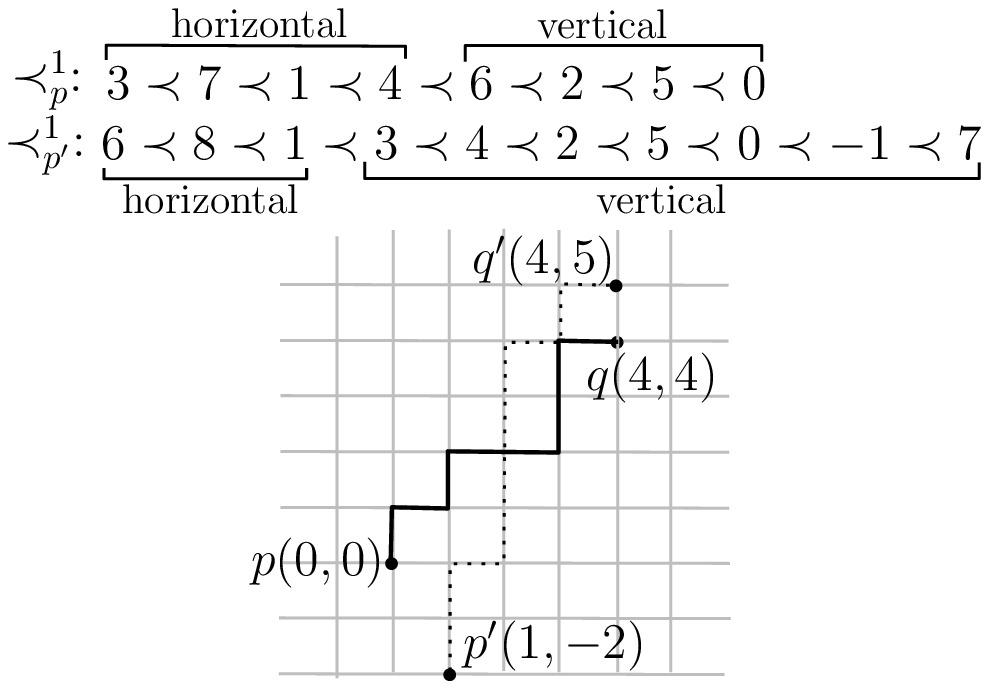}\\
(a) & (b)
\end{tabular}
\caption{(a) A choice of $\prec_p^1$ and $\prec_{p'}^1$ that satisfies (S3). (b) A choice of $\prec_p^1$ and $\prec_{p'}^1$ that does not satisfy (S3).}
\label{fig:charFig}
\end{figure}

We are now ready to give our characterization.  We assume that we are considering segments $R_p(q)$ with nonnegative slope for the majority of this section, and we give set of necessary and sufficient conditions which $\prec^1_{p}$ must satisfy for each $p$ in $\mathbb{Z}^2$. To help explain what must happen we first look at the interaction between first quadrant rays and third quadrant rays. Even though our conditions are only with respect to first-quadrant total orders, it will be useful to describe that our condition is necessary by showing that if the condition is not satisfied then there is some point $q$ such that \textit{any} definition of $\prec^3_q$ would generate third-quadrant segments that violate (S3).

Suppose, we have the total order $\prec^1_{p_1}$ for some point $p_1 \in \mathbb{Z}^2$, and let $p_2 \in Q^1_{p_1}$ and recall that we must have $R_{p_2}(p_1) = R_{p_1}(p_2)$ by property (S2). Now consider how $\prec^3_{p_2}$ must be defined so that $R_{p_2}(p_1) = R_{p_1}(p_2)$. Any integer on which $R_{p_2}(p_1)$ moves horizontally should be smaller than any integer on which the path moves vertically with respect to $\prec^3_{p_2}$, otherwise $R_{p_2}(p_1) \neq R_{p_1}(p_2)$. Motivated by this, we say that an integer on which $R_{p_2}(p_1)$ moves vertically has \textit{priority} over an integer on which it moves horizontally. So, for any two integers $a$ and $b$ such that $a$ has priority over $b$, $(a+1)$ must be larger than $(b+1)$ with respect to $\prec^3_{p_2}$.

Now, suppose we have three points $p_1, p_2, p_3$, where $p_1, p_2 \in Q^3_{p_3}$. $\prec^3_{p_3}$ has a set of priorities induced by $R_{p_1}(p_3)$ and another set of priorities induced by $R_{p_2}(p_3)$. Let $a$ and $b$ be two integers in $I(p_1,p_3) \cap I(p_2,p_3)$. If $a$ has priority over $b$ in $\prec^1_{p_1}$ and $b$ has priority over $a$ in $\prec^1_{p_2}$, then we call this a \textit{conflicting priority}. If we have a conflicting priority then any definition of $\prec^3_{p_3}$ will violate (S2). Indeed, if $(b+1) \prec^3_{p_3} (a+1)$ then this would imply $R_{p_3}(p_2) \neq R_{p_2}(p_3)$, and if $(a+1) \prec^3_{p_3} (b+1)$ then this would imply $R_{p_3}(p_1) \neq R_{p_1}(p_3)$.  Therefore it is necessary to define $\prec^1_{p_1}$ and $\prec^1_{p_2}$ so that there will not be any conflicting priorities for any choice of $p_3 \in Q_{p_1}^1 \cap Q_{p_2}^1$.     

\begin{figure}[htpb]
\centering
\begin{tabular}{c@{\hspace{0.02\linewidth}}c@{\hspace{0.02\linewidth}}c@{\hspace{0.02\linewidth}}c}

\includegraphics[width=1.6in]{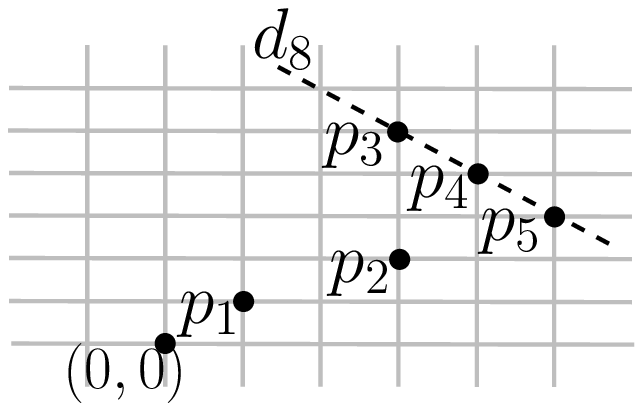} &
\includegraphics[width=1.3in]{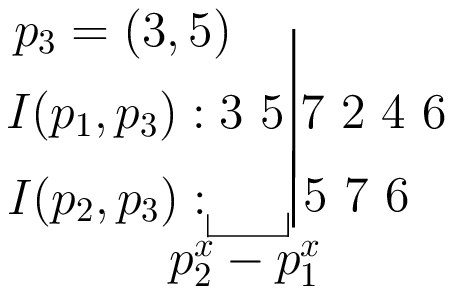} &
\includegraphics[width=1.3in]{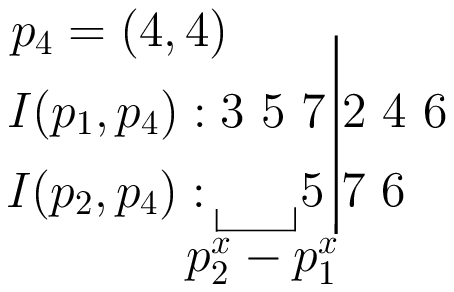} &
\includegraphics[width=1.3in]{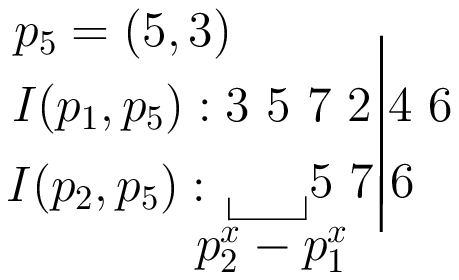} \\
(a) & (b) & (c) & (d)
\end{tabular}
\caption{The layout view of the intervals with $p_1 = (1,1)$ and $p_2 = (3,2)$. (a) The points in the grid. (b) Dividing line for $p_3$. (c) Dividing line for $p_4$. (d) Dividing line for $p_5$.}
\label{fig:layout}
\end{figure}

To help visualize what must happen to avoid these conflicting priorities, we describe a ``layout'' of the integers in the interval. Consider a point $p_3 \in Q^1_{p_1} \cap Q^1_{p_2}$ and the intervals $I(p_1,p_3)$ and $I(p_2,p_3)$ that are used to define the segments $R_{p_1}(p_3)$ and $R_{p_2}(p_3)$ respectively, and without loss of generality assume that $p^x_1 \leq p^x_2$.  We write the intervals in increasing order in a matrix with two rows with $I(p_1,p_3)$ in the top row and $I(p_2,p_3)$ in the bottom row. The first element of $I(p_2,p_3)$ is ``shifted'' to the right $(p^x_2 - p^x_1)$ positions after the first element of $I(p_1,p_3)$. Note that the integers in $I(p_1,p_3)$ and $I(p_2,p_3)$ are determined by the natural total order on the integers, but then are sorted by the total orders $\prec^1_{p_1}$ and $\prec^1_{p_2}$ respectively. The advantage of the layout view is that a single vertical line can break both of the intervals into the horizontal movements portion and vertical movements portion. We call such a line a \textit{dividing line}. The left parts consist of the integers on which the segments make horizontal movements and the right parts consist of the integers on which 
the segments make vertical movements. We define the \textit{antidiagonal} $d_C$ to be the set of all of the points $p=(p^x,p^y)$ in $\mathbb{Z}^2$ such that $(p^x+p^y)=C$. Note that for any two points $q,q'\in Q^1_p \cap d_C$, we have $I(p,q) = I(p,q')$, and if we ``slide'' $q$ up (resp. down) that antidiagonal $d_C$, then the dividing line that corresponds to $q$ moves to the left (resp. to the right).  See Fig.~\ref{fig:layout}. Now, let $a$ and $b$ be two integers in $I(p_1,p_3) \cap I(p_2,p_3)$, and consider these intervals in layout view. Suppose there exists some dividing line $\ell$ such that in $I(p_1,p_3)$ we have $a$ on the left side of $\ell$ and $b$ on the right side of $\ell$, and simultaneously in $I(p_2,p_3)$ we have $b$ on the left side  of $\ell$ and $a$ on the right side of $\ell$. Then we call $\{a,b\}$ a \textit{bad pair}, and we say that $\ell$ \textit{splits} the bad pair. See Fig.~\ref{fig:bad_delete_large} (a). We say total orders $\prec^1_{p_1}$ and $\prec^1_{p_2}$ have a bad pair if there is a $C$ satisfying $C \geq (p^x_1+p^y_1)$ and $C \geq (p^x_2+p^y_2)$ such that the interval $[p^x_1+p^y_1,C]$ sorted by $\prec^1_{p_1}$ and the interval $[p^x_2+p^y_2,C]$ sorted by $\prec^1_{p_2}$ in the layout view have a bad pair. Now we have the following lemma.

\begin{figure}[htpb]
\centering
\begin{tabular}{c@{\hspace{0.01\linewidth}}c@{\hspace{0.01\linewidth}}c}

\includegraphics[width=1.9in]{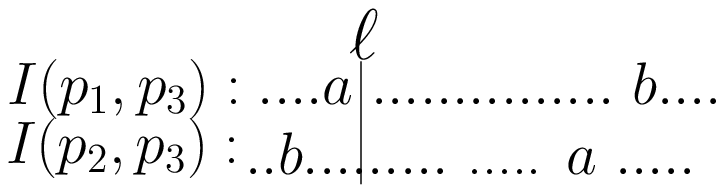} &
\includegraphics[width=1.2in]{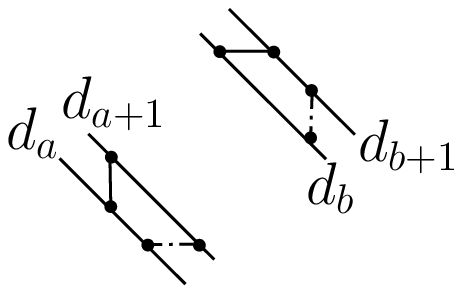}\\
(a) & (b) 
\end{tabular}
\caption{(a) An illustration of a bad pair. (b) An illustration of conflicting priority. }
\label{fig:bad_delete_large}
\end{figure}

\begin{lemma}

If $\prec^1_{p_1}$ and $\prec^1_{p_2}$ have a bad pair, then there exists a $p_3 \in Q^1_{p_1} \cap Q^1_{p_2}$ such that $I(p_1,p_3)$ and $I(p_2,p_3)$ have a bad pair and the dividing line corresponding to $p_3$ splits this bad pair. 
 \label{lem:sup}
\end{lemma}

\begin{proof}
 Suppose that $\prec^1_{p_1}$ and $\prec^1_{p_2}$ have a bad pair.  Let $C$ where $C \geq (p^x_1+p^y_1)$ and $C \geq (p^x_2+p^y_2)$ be such  that there is a bad pair in $[p_1^x+p_1^y,C]$ and $[p_2^x + p_2^y, C]$ in layout view. Let $\{a,b\}$ denote the bad pair in the intervals, and let $\ell$ denote a dividing line that splits that bad pair. Let $p_3$ be a point where $p^x_3+p^y_3=C$ and $\ell$ is the dividing line corresponding with $p_3$. We complete the proof by showing that $p_3 \in Q^1_{p_1} \cap Q^1_{p_2}$. 

Because $\{a,b\}$ is a bad pair, we can assume without loss of generality that $a$ is to the left of $\ell$ and $b$ is to the right of $\ell$ in $I(p_1,p_3)$. This implies that there is at least one horizontal movement and at least one vertical movement to get from $p_1$ to $p_3$ (i.e.,  $p_3^x > p_1^x$ and $p_3^y > p_1^y$).  On the other hand, $a$ is to the right of $\ell$ and $b$ is to the left of $\ell$ in $I(p_2,p_3)$. So we similarly have $p^x_3 > p^x_2$ and $p^y_3 > p^y_2$. So, $p_3 \in Q^1_{p_1} \cap Q^1_{p_2}$, completing the proof.
\end{proof}

The following lemma implies that it is necessary that any pair of first quadrant total orders do not have any bad pairs.

 \begin{lemma}
 There is a point $p_3 \in Q^1_{p_1} \cap Q^1_{p2}$ that has a conflicting priority with respect to $\prec^1_{p_1}$ and $\prec^1_{p_2}$ if and only if there is a bad pair in $\prec^1_{p_1}$ and $\prec^1_{p_2}$. 
    \label{lemma:bad_pair}
  \end{lemma}

\begin{proof}
Assume $\prec^1_{p_1}$ and $\prec^1_{p_2}$ have a bad pair. We will show that there is a point $p_3\in Q^1_{p_1} \cap Q^1_{p_2}$ that has a conflicting priority. 

Let $p_3$ be a point as described in Lemma \ref{lem:sup}, and let $\{a,b\}$ denote the bad pair that $p_3$'s dividing line splits. Without loss of generality, $\prec^3_{p_3}$ must give $(a+1)$ priority over $(b+1)$ with respect to $R_{p_1}(p_3)$ and must give $(b+1)$ priority over $(a+1)$ with respect to $R_{p_2}(p_3)$. Therefore we have a conflicting priority. See Fig.~\ref{fig:bad_delete_large} (b).

Now assume that there is a conflicting priority for $p_3$ with respect to $\prec^1_{p_1}$ and $\prec^1_{p_2}$.  We will complete the proof by showing that $I(p_1,p_3)$ and $I(p_2,p_3)$ must have a bad pair.  Let $a$ and $b$ denote the integers in the conflicting priority, and let $\ell$ denote the dividing line with respect to $p_3$ for $I(p_1,p_3)$ and $I(p_2,p_3)$ in layout view. Then by the definition of conflicting priority we must have $a$ to the left of $\ell$ and $b$ to the right of $\ell$ in one interval, and simultaneously we have $b$ to the left of $\ell$ and $a$ to the right of $\ell$ in the other interval, forming a bad pair.  Since $I(p_1,p_3)$ and $I(p_2,p_3)$ have a bad pair, we have that $\prec^1_{p_1}$ and $\prec^1_{p_2}$ have a bad pair. 
\end{proof}

Lemma \ref{lemma:bad_pair} implies that it is necessary for any two points $p_i$ and $p_j$ that $\prec^1_{p_i}$ and $\prec^1_{p_j}$ do not have 
a bad pair. We now show that 
this condition is also sufficient.  

 \begin{lemma}
 If all pairs of total orders have no bad pairs, then the line segments will satisfy properties (S1)-(S5). 
    \label{lemma:sufficient}
  \end{lemma} 

\begin{proof}
It is easy to see that (S1), (S2), (S4) and (S5) are automatically satisfied by construction, and it is only (S3) that we need to prove.  We first will show that a segment with nonnegative slope and a segment with non-positive slope will always satisfy (S3).  To see this, consider a nonnegative line segment $R_{p'}(q')$ and a non-positive line segment $R_p(q)$.  If they violate $(S3)$ then there must be a witness $(t_1,t_2)$ in $R_p(q) \cap R_{p'}(q')$ for two points $t_1=(t^x_1, t^y_1)$ and $t_2=(t^x_2, t^y_2)$ such that $t^x_1 \neq t^x_2$ and $t^y_1 \neq t^y_2$.  But we will show that for any two points $r_1$ and $r_2$ that satisfy $r^x_1 \neq r^x_2$ and $r^y_1 \neq r^y_2$, it cannot be that $r_1$ and $r_2$ are in both segments. Without loss of generality, assume that $r^x_1 < r^x_2$ and $r_1 \in R_{p'}(q') \cap R_p(q)$. 

All points after $r_1$ in $R_{p'}(q')$ have $y$-coordinate at least $r^y_1$, and all points after $r_1$ in $R_p(q)$ have $y$-coordinate at most $r^y_1$.  Therefore if there is a point $z$ that comes after $r_1$ in $R_{p'}(q') \cap R_p(q)$ then it must satisfy $z^y = r_1^y$.  This implies that if both segments contain $r_1$ and $R_{p}(q)$ contains $r_2$ then $R_{p'}(q')$ cannot contain $r_2$. Thus $R_{p}(q)$ and $R_{p'}(q')$ do not have a witness and do not violate $(S3)$.

Now without loss of generality, consider segments $R_{p_1}(q_1)$ and $R_{p_2}(q_2)$ with nonnegative slope.  In order to violate (S3), there must be a witness $(t_1,t_2)$ to the violation of (S3). Suppose we have such a witness, and consider the subsegments $R_{p_1}(t_2)$ and $R_{p_2}(t_2)$.  If we consider the intervals $I(p_1,t_2)$ and $I(p_2,t_2)$ in layout view, we can see that the dividing line corresponding to $t_2$ will split the bad pair $\{t_1^x+t_1^y,t_2^x+t_2^y-1\}$. Therefore if there are no bad pairs, then there cannot be a witness to the violation of (S3).
\end{proof}

Combining Lemma \ref{lemma:bad_pair} and \ref{lemma:sufficient}, we get the following Theorem \ref{thm:main}.  
\begin{theorem}
\label{thm:main}
A system of nonnegative sloped line segments in $\mathbb{Z}^2$ is a CDS if and only if we have a total order for the first quadrants for each point such that each pair of total orders have no bad pairs and the third quadrant segments are induced by the corresponding first quadrant segments.    
\end{theorem}

\section{Luby and Christ et al.~in the Context of Our Characterization}
\label{sec:luby_christ}
In an attempt to tie together some of the previous works on CDSes, we now analyze the work of Luby \cite{Luby88} and Christ et al.~\cite{ChristPS12} in the context of our characterization.  Chun et al.~and Christ et al.~were not aware of Luby's work when publishing \cite{ChristPS12} and \cite{ChunKNT09}, although Christ gives a comparison of his work with that of Luby in his thesis \cite{Christ2011}.  

We will first provide an analysis relating smooth grid geometries given by Luby \cite{Luby88}.  To do so, we need the following definition.  Consider any two points $p_1$ and $p_2$ with first quadrant total orders $\prec^1_{p_1}$ and $\prec^1_{p_2}$.  We say that $\prec^1_{p_1}$ and $\prec^1_{p_2}$ are \textit{in agreement} if $a \prec^1_{p_1} b$ if and only if $a \prec^1_{p_2} b$ for every pair of integers $a$ and $b$ such that antidiagonals $d_a$ and $d_b$ intersect $Q^1_{p_1} \cap Q^1_{p_2}$.  Intuitively, $\prec^1_{p_1}$ and $\prec^1_{p_2}$ are in agreement if they are the same ordering when considering antidiagonals intersecting both first quadrants.  We now prove the following lemma about smooth grid geometries.  An equivalent lemma was proved by Christ \cite{Christ2011} using a different proof technique.  We argue the lemma for segments with nonnegative slope, but an equivalent argument holds for segments with negative slope.  

\begin{lemma}
 A CDS is a smooth grid geometry if and only if $\prec^1_p$ and $\prec^1_q$ are in agreement for any pair of points $p,q \in \mathbb{Z}^2$.
    \label{lemma:smoothness}
  \end{lemma}  

\begin{proof}
First, we prove that if the total orders of all pairs of points in $\mathbb{Z}^2$ are in agreement, then the induced CDS is smooth by proving the contrapositive. Suppose, a CDS is not smooth. Then there are some line segments $R_{p_1}(q_1)$ and $R_{p_2}(q_2)$ that are not smooth. By the definition of smoothness, we have antidiagonals $d_E, d_F, d_G$ and that intersect both $R_{p_1}(q_1)$ and $R_{p_2}(q_2)$ such that the antidiagonal distance function dist() is not monotonically increasing or decreasing along these antidiagonals. Now, without loss of generality, assume that the distance between them increases from $d_E$ to $d_F$ and decreases from $d_F$ to $d_G$. See Fig.~\ref{fig:smoothness} (a). 

\begin{figure}[htpb]
\centering
\begin{tabular}{c@{\hspace{0.01\linewidth}}c@{\hspace{0.01\linewidth}}c@{\hspace{0.01\linewidth}}c}

\includegraphics[width=1.5in]{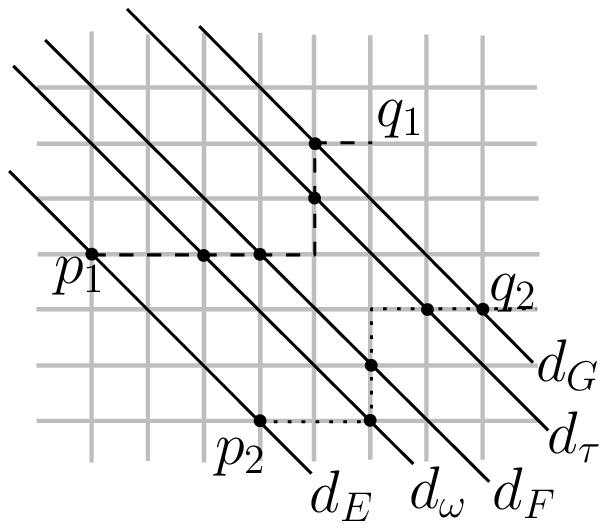} &
\includegraphics[width=1.5in]{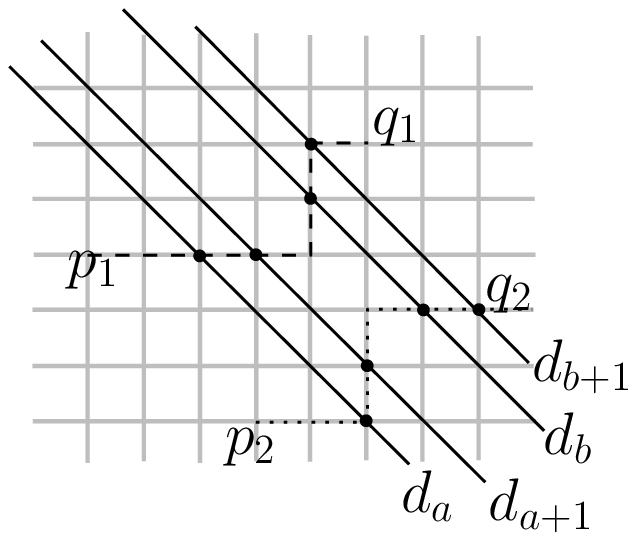}\\

(a) & (b)
\end{tabular}
\caption{An illustration of Lemma \ref{lemma:smoothness}.}
\label{fig:smoothness}
\end{figure}

Now, we claim that the total orders $\prec_{p_1}$ and $\prec_{p_2}$ are not in agreement.  We start from $d_E$ to check each antidiagonal between $d_E$ and $d_F$ to find the first antidiagonal where $R_{p_1}(q_1)$ goes horizontally and $R_{p_2}(q_2)$ goes vertically. Clearly such an antidiagonal exists since the antidiagonal distance at $d_F$ is smaller than at $d_E$, and let $d_\omega$ be such an antidiagonal. Similarly, we start from $d_F$ to find an antidiagonal between $d_F$ and $d_G$ to find the first antidiagonal where $R_{p_1}(q_1)$ goes vertically and $R_{p_2}(q_2)$ goes horizontally. Let $d_\tau$ be such an antidiagonal. As the line segment $R_{p_1}(q_1)$ goes horizontally on $d_\omega$ and vertically on $d_\tau$, we have $\omega \prec_{p_1} \tau$. On the other hand, $R_{p_2}(q_2)$ goes vertically on $d_\omega$ and horizontally on $d_\tau$, so we have  $\tau \prec_{p_2} \omega$. So, the two total orders do not agree on $\omega$ and $\tau$. It follows that if all of the total orders are in agreement, then the CDS is smooth.

Now, we prove that if two line segments do not agree on a pair of elements, then the CDS defined by those segments is not smooth. Suppose, we have two points $p_1$ and $p_2$ and suppose the total orders $\prec_{p_1}$ and $\prec_{p_2}$ assigned to them disagree on some numbers $a$ and $b$. Now, we show that there exists two points $q_1 \in Q^1_{p_1}$ and $q_2 \in Q^1_{p_2}$ so that $R_{p_1}(q_1)$ and $R_{p_2}(q_2)$ are not smooth. We choose a point $q_1$ such that $(q_1^x + q_1^y) > \text{max}(a,b)$ and the dividing line corresponding to $q_1$ is between $a$ and $b$ in layout view for $I(p_1,q_1)$ when sorted by $\prec_{p_1}$.  Without loss of generality suppose we have $a \prec_{p_1} b$, then $R_{p_1}(q_1)$ goes horizontally on $a$ and goes vertically on $b$. We choose $q_2$ similarly, that is the dividing line corresponding to $q_2$ is between $a$ and $b$ in layout view for $I(p_2,q_2)$ when sorted by $\prec_{p_2}$.  Since $\prec_{p_1}$ and $\prec_{p_2}$ disagree on $a$ and $b$, we must have $b \prec_{p_2} a$, and therefore we have $R_{p_2}(q_2)$ goes horizontally on $b$ and goes vertically on $a$. Now, consider the antidiagonal $d_a$ and $d_b$ which intersects $R_{p_1}(q_1)$ and $R_{p_2}(q_2)$. See Fig.~\ref{fig:smoothness} (b). As $R_{p_1}(q_1)$ goes horizontally on $a$ and $R_{p_2}(q_2)$ goes vertically on $a$, their distance on $d_{a+1}$ is decreased and also $R_{p_2}(q_2)$ goes horizontally on $b$ and $R_{p_1}(q_1)$ goes vertically on $b$, their distance on $d_{b+1}$ is increased. So, the distance between the line segments is not monotonically increasing or decreasing.   It follows that the corresponding CDS is not smooth. 
\end{proof}

We now turn our attention to analyzing the work of Christ et al.~\cite{ChristPS12} in the context of bad pairs.  They give two methods for choosing total orders to construct a CDS.  The first method is to assign total orders to points so that all pairs of total orders are in agreement (e.g., assigning the same total order to all points).  Note that if two total orders have a bad pair, then there necessarily has to be two integers $a$ and $b$ such that one total order has $a \prec b$ while another has $b \prec a$.  But this clearly cannot happen if all total orders are in agreement.  Therefore there are no bad pairs and by Theorem \ref{thm:main} it is a CDS.  

They also give an example of a CDS constructed using total orders that are not in agreement.  Specifically, a point's total order is chosen depending on if it is above or below the $x$-axis. Because of the special role of the $x$-axis, this example is called the \textit{waterline example}. In the waterline example, every point $p$ such that $p^y \geq 0$ uses the natural total order, that is $\prec_p^1 = (p^x + p^y) \prec (p^x + p^y +1) \prec \cdots \prec (+\infty)$.  For points $p$ such that $p^y < 0$, the total order is a function of its $x$-coordinate $p^x$.  Specifically, we have $\prec_p^1 = (p^x \prec p^x + 1) \prec \cdots \prec (+\infty) \prec (p^x -1) \prec (p^x - 2) \prec \cdots \prec (-\infty)$.  After giving this definition, Christ et al.~point out that it is easy to see that the segments form a CDS.  We give a formal proof using our characterization.

\begin{lemma}
\label{lemma:waterline}
The waterline example is a CDS.
\end{lemma}

\begin{proof}
The waterline example clearly satisfies all properties other than (S3).  We show that it also satisfies (S3) by showing that no pair of total orders from the waterline example have a bad pair. Suppose $p_1=(p_1^x,p_1^y)$ and $p_2=(p_2^x,p_2^y)$ are such that $p_1^y \geq 0$ and $p_2^y \geq 0$.  Then they use the same total order and they clearly don't have a bad pair.  If $p_1^y < 0$, $p_2^y < 0$, and $p_1^x = p_2^x$ then they also use the same total order and therefore will not have a bad pair.  So now suppose that $p_1^y < 0$, $p_2^y < 0$, and $p_1^x \neq p_2^x$.  Without loss of generality, assume that $p_1^x < p_2^x$, and moreover let $p_2^x = p_1^x + \tau$ for some integer $\tau > 0$.  Let $q$ be any point in $Q^1_{p_1} \cap Q^1_{p_2}$, and let $C = q^x + q^y$. We will now show that $I(p_1,q)$ and $I(p_2,q)$ do not have a bad pair.

\begin{figure}[htpb]
\centering
\includegraphics[width=4.5in]{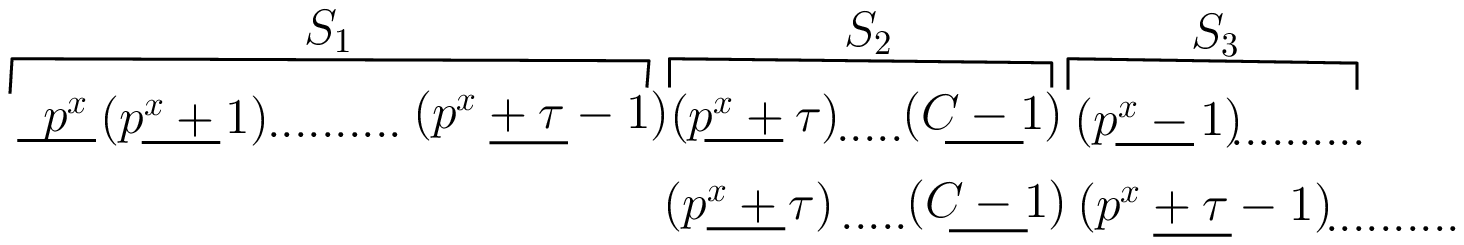}
\caption{An illustration of that the waterline example does not have a bad pair.}
\label{fig:water1}
\end{figure} 

Consider the intervals in layout view.  See Fig.~\ref{fig:water1}.  We partition the layout view into three sections $S_1, S_2$, and $S_3$.  $S_1$ consists of the smallest $\tau$ integers of $I(p_1,q)$ according to $\prec_{p_1}^1$, $S_2$ consists of the layout view containing $x+\tau$ through $C-1$ in both intervals, and $S_3$ consists of the integers to the right of $C-1$ in both intervals.  We also partition the interval $(-\infty, C-1]$ into three sub-intervals $I_1 = (-\infty, x-1]$, $I_2 = [x, x+\tau - 1]$, and $I_3 = [x+\tau, C-1]$ (note that these intervals are with respect to the natural total order and not $\prec_{p_1}^1$ and $\prec_{p_2}^1$).  We now argue that there are no bad pairs $\{a,b\}$ in $I(p_1,q)$ and $I(p_2,q)$.  First let $a$ be any integer in $I_3$, and let $b$ be any integer in $(-\infty, C-1]$.  Note that $a$ is in $S_2$, and it is in the same position in both of the intervals.  Even if $\prec_{p_1}^1$ and $\prec_{p_2}^1$, disagree on the ordering of $a$ and $b$, clearly there cannot be a dividing line that splits them (because $a$ is in the same position).  So we now suppose that $a$ and $b$ are not in $I_3$.  Next suppose $a$ and $b$ are both in $I_1$.  In this case both $a$ and $b$ are in $S_3$ in both intervals, and both total orders will order $a$ and $b$ in the same way and therefore they cannot be a bad pair.  If $a$ and $b$ are both in $I_2$ then the total orders will disagree on their ordering, but $a$ and $b$ will both be in $S_1$ for $I(p_1,q)$ and they will both be in $S_3$ for $I(p_2,q)$.  It clearly follows that they cannot be a bad pair.  Finally suppose that $a \in I_1$ and $b \in I_2$.  We then have that both total orders will order $a$ and $b$ the same way, and therefore they cannot be a bad pair.  It follows that $\prec_{p_1}^1$ and $\prec_{p_2}^1$ do not have a bad pair.

\begin{figure}[htpb]
\centering
\includegraphics[width=4.5in]{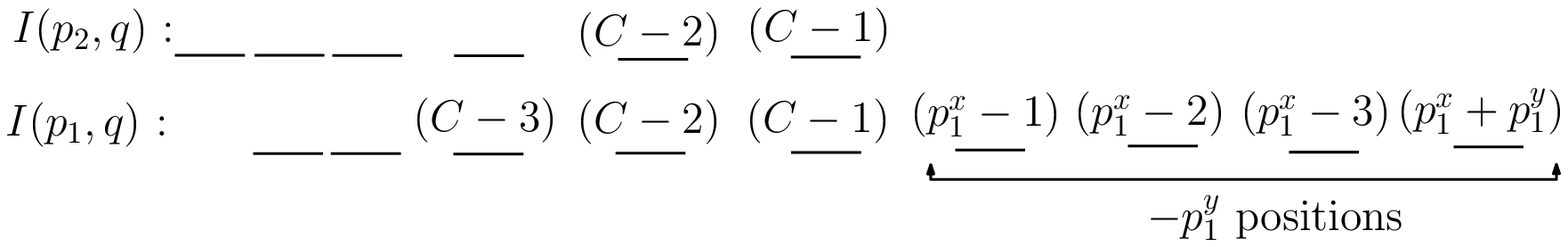}
\caption{An illustration of that the waterline example does not have a bad pair.}
\label{fig:waterline}
\end{figure}  

Now suppose that $p_1^y < 0$ and $p_2^y \geq 0$.  Again let $q$ be any point in $Q^1_{p_1} \cap Q^1_{p_2}$, and let $C = q^x + q^y$.  We complete the proof by showing that $I(p_1,q)$ and $I(p_2,q)$ do not have a bad pair.  Recall that $\prec^1_{p_2}$ is simply the natural total order on the interval $I(p_2,q)$.   Note that the subinterval $[p_1^x,C-1] \subseteq I(p_1,q)$ sorted by $\prec^1_{p_1}$ also results in the natural total order on this subinterval.  Therefore if we let $a$ and 
$b$ be any two integers in $[p_1^x,C-1]$ then we have that $\{a,b\}$ cannot be a bad pair as both total orders will agree on the relative ordering of $a$ and $b$. This implies that if there does exist a bad pair $\{a,b\}$ then at least one of the two integers needs to be in the subinterval $[p_1^x + p_1^y, p_1^x - 1]$, so without loss of generality assume that $a \in [p_1^x + p_1^y, p_1^x - 1]$.  We will prove there are no bad pairs $\{a,b\}$ by considering two different cases: (1) when $b \in [p_1^x + p_1^y, p_1^x - 1]$ and (2) when $b \in [p_1^x,C-1]$.

First suppose that $b \in [p_1^x + p_1^y, p_1^x - 1]$.  Note that because of the $y$-coordinates of $p_1, p_2$, and $q$, it is the case that $R_{p_1}(q)$ must make at least $-p_1^y$ more vertical movements to reach $q$ than $R_{p_2}(q)$.  Therefore when considering $I(p_1,q)$ and $I(p_2,q)$ in layout view, we will have at least $-p_1^y$ integers from $I(p_1,q)$ that are to the right of the largest integers in $I(p_2,q)$.  Also note that in this case, $a$ and $b$ are both amongst the $-p_1^y$ largest elements of $I(p_1,q)$ according to $\prec^1_{p_1}$, and therefore will be positioned to the right of the largest element of $I(p_2,q)$ in layout view.  It immediately follows that $\{a,b\}$ cannot be a bad pair.  See Fig.~\ref{fig:waterline}.

Now suppose that $b \in [p_1^x,C-1]$.  Note that in this case $\prec^1_{p_1}$ has $a$ being larger than $b$, and $\prec^1_{p_2}$ has $b$ as being larger than $a$. Similarly to last time, in the layout view we have the position of $a$ in $I(p_1,q)$ is to the right of the largest element of $I(p_2,q)$, and therefore in this case $\{a,b\}$ is a bad pair if and only if the position of $b$ in the layout view of $I(p_1,q)$ is strictly to the left of the position of $b$ in the layout view of $I(p_2,q)$.  Note that $(C-1)$ is the largest element of $I(p_2,q)$ and is exactly one position to the left of $(p_1^x-1)$ in $I(p_1,q)$, and this implies that the position of $(C-1)$ in $I(p_1,q)$ is either to the right of its position in $I(p_2,q)$ or in the same position.  Since both total orders use the natural total order for all elements positioned to the left of $(q^x + q^y-1)$, it follows that the position of $b$ in $I(p_1,q)$ is either to the right of its position in $I(p_2,q)$ or in the same position.  It then follows that $\{a,b\}$ is not a bad pair.  See Fig.~\ref{fig:waterline}.
\end{proof}

\begin{figure}[htpb]
\centering
\begin{tabular}{c@{\hspace{0.01\linewidth}}c@{\hspace{0.01\linewidth}}c}

\includegraphics[width=1.4in]{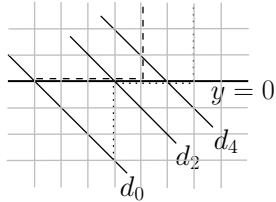}\\

\end{tabular}
\caption{An illustration of that the waterline example is not smooth.}
\label{fig:bad_delete_large1}
\end{figure}

Note that the waterline example is a CDS that is not smooth. See Fig.~\ref{fig:bad_delete_large1}. The dashed segment is $R_p(q)$ with $p = (0,0)$ and $q = (4,3)$ and the dotted segment is $R_{p'}(q')$ with $p' = (3,-3)$ and $q' = (6,3)$.  If we let $z = (3,0)$, then we have that dist($R_p(q),R_{p'}(q'),p) = -3$, dist($R_p(q),R_{p'}(q'),z) = 0$, and dist($R_p(q),R_{p'}(q'),q) = -2$.  By definition, these segments are not smooth and therefore the waterline example is not a smooth grid geometry. 

\section{Acknowledgement}
We would like to thank Dr.~Xiaodong Wu for introducing the problem to us, and Dr.~Wu and Dr.~Kasturi Varadarajan for some valuable discussions.\\

\appendix



\begin{thebibliography}{10}

\bibitem{Andres03}
Eric Andres.
\newblock Discrete linear objects in dimension n: the standard model.
\newblock {\em Graphical Models}, 65(1-3):92--111, 2003.

\bibitem{BertheL11}
Val{\'{e}}rie Berth{\'{e}} and S{\'{e}}bastien Labb{\'{e}}.
\newblock An arithmetic and combinatorial approach to three-dimensional
  discrete lines.
\newblock In {\em Discrete Geometry for Computer Imagery - 16th {IAPR}
  International Conference, Nancy, France, April 6-8}, pages 47--58, 2011.

\bibitem{Christ2011}
Tobias Christ.
\newblock {\em {D}iscrete {D}escriptions of {G}eometric {O}bjects}.
\newblock PhD thesis, Z{\"{u}}rich, 2011.

\bibitem{ChristPS12}
Tobias Christ, D{\"o}m{\"o}t{\"o}r P{\'a}lv{\"o}lgyi, and Milos Stojakovic.
\newblock Consistent digital line segments.
\newblock {\em Discrete {\&} Computational Geometry}, 47(4):691--710, 2012.

\bibitem{ChunKNT09}
Jinhee Chun, Matias Korman, Martin N{\"o}llenburg, and Takeshi Tokuyama.
\newblock Consistent digital rays.
\newblock {\em Discrete {\&} Computational Geometry}, 42(3):359--378, 2009.

\bibitem{Cohen-OrK97}
Daniel Cohen{-}Or and Arie~E. Kaufman.
\newblock 3d line voxelization and connectivity control.
\newblock {\em {IEEE} Computer Graphics and Applications}, 17(6):80--87, 1997.

\bibitem{Eckhardt00}
Ulrich Eckhardt.
\newblock Digital lines and digital convexity.
\newblock In Gilles Bertrand, Atsushi Imiya, and Reinhard Klette, editors, {\em
  Digital and Image Geometry}, volume 2243 of {\em Lecture Notes in Computer
  Science}, pages 209--228. Springer, 2000.

\bibitem{Franklin85}
Wm.~Randolph Franklin.
\newblock Problems with raster graphics algorithm.
\newblock In Frans J.~Peters Laurens R. A.~Kessener and Marloes L.~P. van
  Lierop, editors, {\em Data Structures for Raster Graphics, Steensel,
  Netherlands}, 1985.

\bibitem{GoodrichGHT97}
Michael~T. Goodrich, Leonidas~J. Guibas, John Hershberger, and Paul~J.
  Tanenbaum.
\newblock Snap rounding line segments efficiently in two and three dimensions.
\newblock In {\em Symposium on Computational Geometry}, pages 284--293, 1997.

\bibitem{GreeneY86}
Daniel~H. Greene and Frances~F. Yao.
\newblock Finite-resolution computational geometry.
\newblock In {\em 27th Annual Symposium on Foundations of Computer Science,
  Toronto, Canada, 27-29 October}, pages 143--152. {IEEE} Computer Society,
  1986.

\bibitem{KletteR04}
Reinhard Klette and Azriel Rosenfeld.
\newblock Digital straightness - a review.
\newblock {\em Discrete Applied Mathematics}, 139(1-3):197--230, 2004.

\bibitem{LachaudS13}
Jacques-Olivier Lachaud and Mouhammad Said.
\newblock Two efficient algorithms for computing the characteristics of a
  subsegment of a digital straight line.
\newblock {\em Discrete Applied Mathematics}, 161(15):2293--2315, 2013.

\bibitem{Luby88}
Michael~G. Luby.
\newblock Grid geometries which preserve properties of euclidean geometry: A
  study of graphics line drawing algorithms.
\newblock In Rae~A. Earnshaw, editor, {\em Theoretical Foundations of Computer
  Graphics and CAD}, volume~40, pages 397--432, 1988.

\bibitem{SivignonDC04}
Isabelle Sivignon, Florent Dupont, and Jean-Marc Chassery.
\newblock Digital intersections: minimal carrier, connectivity, and periodicity
  properties.
\newblock {\em Graphical Models}, 66(4):226--244, 2004.

\bibitem{Sugihara01}
Kokichi Sugihara.
\newblock Robust geometric computation based on topological consistency.
\newblock In {\em International Conference on Computational Science (1)}, pages
  12--26, 2001.

\bibitem{LieropOW88}
M.~L.~P. van Lierop, Cornelius W. A.~M. van Overveld, and H.~M.~M. van~de
  Wetering.
\newblock Line rasterization algorithms that satisfy the subset line property.
\newblock {\em Computer Vision, Graphics, and Image Processing},
  41(2):210--228, 1988.

\end{thebibliography}
\end{document}